\documentclass[10pt]{article}

\usepackage[margin=1in]{geometry}
\usepackage{natbib}
\usepackage{parskip}

\usepackage{algorithm, algorithmicx, algpseudocode}
\usepackage{amssymb}
\usepackage{amsmath}
\usepackage{graphicx}
\usepackage{mathtools}
\usepackage{needspace} 
\usepackage{multicol}
\usepackage{listings}
\usepackage{amsthm}
\usepackage[shortlabels]{enumitem}
\usepackage[italicdiff]{physics}
\usepackage{cancel}
\usepackage[dvipsnames]{xcolor}
\usepackage[table]{xcolor}
\usepackage{verbatim}
\usepackage{comment}
\usepackage{array}
\usepackage{tikz-cd}
\usepackage{import}
\usepackage{booktabs}
\usepackage{comment}
\usepackage{array}
\usepackage[colorlinks, linkcolor=blue, citecolor=ForestGreen]{hyperref}
\usepackage{forest}
\usepackage{float}
\usepackage{xspace}
\usepackage{aliascnt}
\usepackage{wrapfig}
\usepackage{adjustbox}
\usepackage{bm}
\usepackage{bbm}
\usepackage{complexity}
\usepackage{makecell}
\usepackage{nicefrac}
\usepackage{enumitem}
\usepackage{comment}
\usepackage{amsthm}
\usepackage{bbm}
\usepackage{physics}
\usepackage{amssymb}
\usepackage{booktabs}
\usepackage{diagbox}
\usepackage{mathtools}

\usepackage{amsmath}
\usetikzlibrary{arrows.meta,decorations.pathreplacing}
\usetikzlibrary{calc, positioning, external, arrows.meta}

\allowdisplaybreaks

\DeclareMathOperator*{\argmax}{argmax}

\let\cite\citep

\newfloat{protocol}{tbp}{lop}
\floatname{protocol}{Protocol}

\usepackage[most]{tcolorbox}
\tcbuselibrary{listingsutf8}

\tcbset{
  myblockstyle/.style={
    enhanced,
    boxrule=0.4pt,
    colback=#1!10,
    colframe=#1!80!black,
    boxsep=2pt,
    left=-2.5pt,
    right=2pt,
    top=0pt,
    bottom=0pt,
    before skip=2pt,
    after skip=2pt,
  }
}

\newtcbox{\tcbinlinebox}[1][]{nobeforeafter,
  on line,
  colback=#1!10,
  colframe=#1!80!black,
  boxrule=0.8pt,
  arc=2pt,
  boxsep=0pt,
  left=2pt, right=2pt, top=2pt, bottom=2pt,
  enhanced,
  valign=center
}
\usepackage[most]{tcolorbox}

\definecolor{partialteal}{RGB}{82,110,150}
\definecolor{fullpurple}{RGB}{150,95,110}

\newtheorem{theorem}{Theorem}
\numberwithin{theorem}{section}
\newtheorem*{theorem*}{Theorem}

\newtheorem{lemma}[theorem]{Lemma}

\theoremstyle{definition}

\usepackage{thm-restate}

\newcommand{\ind}[1]{\mathbbm{1}\qty{#1}}

\algnewcommand\algorithmicinput{\textbf{Input:}}
\algnewcommand\Input{\item[\algorithmicinput]}

\newenvironment{protocolalgorithm}
  {\floatname{algorithm}{Protocol}\begin{algorithm}}
  {\end{algorithm}\floatname{algorithm}{Algorithm}}

\def\gG{{\mathcal{G}}}
\def\gH{{\mathcal{H}}}

\def\gO{{\mathcal{O}}}

\def\gU{{\mathcal{U}}}

\usepackage{amsmath}
\usetikzlibrary{arrows.meta,decorations.pathreplacing}

\allowdisplaybreaks

\usepackage{pifont}

\newcommand{\xmark}{\ding{55}}

\makeatletter
\renewcommand\thanks[1]{%
  \footnotemark%
  \protected@xdef\@thanks{\@thanks
    \protect\footnotetext[\the\c@footnote]{#1}}%
}
\makeatother

\title{Optimal Regret for Online Market Making \\ with Limit Order Book}
\author{
    Maria Elena Vischi\\
    \texttt{mariaelena.vischi@polimi.it} \\
    Politecnico di Milano 
\and
    Francesco Emanuele Stradi \\\texttt{francescoemanuele.stradi@polimi.it} \\
    Politecnico di Milano 
\and 
    Alberto Marchesi\\
    \texttt{alberto.marchesi@polimi.it} \\
    Politecnico di Milano 
}
\date{\today}

\begin{document}

\maketitle

\begin{abstract}
We study online learning in  market making, where, at each round, a market maker posts bid and ask prices before observing the market price and the private valuation of an incoming trader. In this setting, \citet{Maran2026OnlineMM} introduce a feedback model motivated by limit order books, in which the trader’s valuation is revealed only if no transaction occurs. Assuming that trader valuations are drawn i.i.d.\ from an unknown distribution while market prices are chosen adversarially, they establish an expected regret bound of $\widetilde{\mathcal{O}}(T^{2/3})$. In this work, we improve upon this guarantee by establishing a high-probability regret bound of $\widetilde{\mathcal{O}}(\sqrt{T})$. As a warm-up, we first consider the full-feedback setting. We introduce a discretization of the bid--ask space based on two coupled grids and combine it with Hedge to achieve the desired regret rate. Building on these ideas, we then address the substantially weaker feedback induced by a limit order book and develop an algorithm that achieves the same guarantee. Finally, we investigate the limits of learnability in fully adversarial environments, where the valuations may vary arbitrarily as well. Perhaps surprisingly, we show that when both market prices and trader valuations are chosen adversarially, sublinear regret is impossible even under full feedback, thereby motivating our stochastic assumption on the valuations.
\end{abstract}

\section{Introduction}

Market makers provide liquidity to financial markets by posting prices at which they are willing to buy and sell an asset~\citep{amihud1986asset}. They aim to profit from the bid--ask spread while managing the risks of price changes and trading with better-informed participants, a central theme in the economics literature on market making. Today, market making is a core activity of quantitative trading firms such as Jane Street and Citadel Securities, which use mathematical models and automated trading systems to provide liquidity across financial markets \citep{BIAIS2005217}. Designing effective strategies for choosing and updating these quotes is therefore an important problem in algorithmic trading~\citep{glosten1988estimating,madhavan2000market}.
A central challenge in designing market-making strategies is that quotes must be chosen under uncertainty about both the market price of the asset and the willingness of incoming traders to transact. Moreover, the available information depends on the quotes themselves, linking current trading decisions to future learning opportunities. Online learning provides a framework for studying this interaction through a sequential decision-making problem. At each round, a market maker posts a bid and an ask before observing the market price and the incoming trader's private valuation. Thus, the performance is measured by regret against the best fixed bid--ask pair in hindsight.

\cite{CesaBianchi2024MarketMW} study this problem under two-bit feedback, where the learner observes the market price and whether the trader buys, sells, or does not trade, while the trader's valuation remains hidden. They establish tight regret guarantees under several assumptions on market prices and valuations. In particular, the minimax regret is of order $T^{2/3}$ when market prices and trader valuations are i.i.d.\ and independent of each other. Their lower bound holds even under Lipschitz assumptions on the underlying distributions, showing that the difficulty persists in regular stochastic environments. They also establish impossibility results when suitable assumptions on regularity and independence are absent. These results highlight the role of the information available to the market maker in determining learnability.

More recently, \citet{Maran2026OnlineMM} introduce a richer feedback model motivated by limit order books. In this model, the trader's valuation remains hidden when a transaction occurs but is revealed when no transaction takes place. Thus, the posted quotes determine both the reward and the information available for future decisions. Under i.i.d.\ trader valuations, they obtain high-probability regret bounds of $\widetilde{\mathcal{O}}(\sqrt{T})$ for independent stochastic market prices and extend these guarantees to certain mean-reverting price processes, without imposing smoothness on the valuation distribution. Differently, when the market prices are chosen by an oblivious adversary, while the valuations remain stochastic and stationary, they establish an expected regret bound of $\widetilde{\mathcal{O}}(T^{2/3})$. 

This leaves a gap between the guarantees available for stochastic and adversarial market prices under limit order book feedback. Is the slower rate an unavoidable consequence of arbitrary market prices, or can the information revealed by the order book still support a $\widetilde{\mathcal{O}}(\sqrt{T})$ guarantee? Moreover, both guarantees rely on the assumption that trader valuations are drawn from a fixed distribution. It is therefore natural to ask whether learning remains possible when this assumption is removed. Specifically, in this work we address the following two research questions:

\begin{quote}
\centering
\emph{$(i)$ Can we achieve $\widetilde{\mathcal{O}}(\sqrt{T})$ regret under limit order book feedback with adversarial market prices? \\
$(ii)$ Can sublinear regret also be achieved when both market prices and trader valuations are adversarial?}
\end{quote}

We answer the first question affirmatively and the second negatively. When trader valuations are drawn i.i.d.\ from an unknown distribution, we establish a high-probability regret bound of $\widetilde{\mathcal{O}}(\sqrt{T})$ while allowing market prices to remain arbitrary. Finally, we show that sublinear regret is impossible in fully adversarial environments, even under full feedback.

\subsection{Original Contributions}

We study online market making under different assumptions on the environment and the feedback available to the learner. Our main focus is the setting in which market prices are chosen by an adaptive adversary, while private valuations are drawn from an arbitrary distribution. We establish optimal regret rates under both full feedback and limit order book feedback, and show that learning becomes impossible when private valuations are also chosen adversarially. Our contributions can be summarized as follows:
\begin{itemize}
    \item \textbf{Optimal regret under full feedback.}
    As a warm-up, we establish a $\widetilde{\gO}(\sqrt{T})$ regret bound under full feedback. Our algorithm combines two Hedge instances, one for bid prices and one for ask prices. A key observation is that swapping the sampled quotes whenever they violate the feasibility constraint never decreases utility, allowing us to analyze the two instances separately and transfer their guarantees to the actions actually played.
    \item \textbf{Optimal high-probability regret under limit order book.}
    For the feedback model introduced by~\citet{Maran2026OnlineMM}, we establish a $\widetilde{\gO}(\sqrt{T})$ regret bound holding with high probability, even against an adaptive adversary. This improves upon the previous $\widetilde{\gO}(T^{2/3})$ bound, which held in expectation and assumed an oblivious adversary. Our analysis exploits the feedback structure. Specifically, each round allows the learner to reconstruct losses for the interval between the sampled quotes. Using this additional information, we achieve the same regret rate as under full feedback.

    \item \textbf{Impossibility in a fully adversarial environment.}
    When both market prices and private valuations are chosen adversarially, we prove a linear regret lower bound. This impossibility result holds even under full feedback, showing that observing all parameters at the end of each round is insufficient to guarantee sublinear regret.
\end{itemize}

The $\tilde{\gO}(\sqrt{T})$ upper bounds are optimal up to logarithmic factors, matching the standard $\Omega(\sqrt{T})$ lower bound for the multi-armed bandit problem~\citep{lattimore2020bandit}. Together, our results clarify how the assumptions on the environment and the available feedback affect the learnability of market making. 

Table~\ref{table-comparison} summarizes our guarantees and compares them with previous work.

\begin{table}[t]
	\caption{Comparison between our results and the state-of-the-art.}
	\label{table-comparison}
	\centering
	\small
	{\renewcommand{\arraystretch}{1.25}
		\setlength{\tabcolsep}{8pt}
		\begin{tabular}{c c | c c}
			\toprule
            \textbf{Feedback}
			& \textbf{$(M_t,V_t)$}
			&\citep{Maran2026OnlineMM}
            &\textbf{Our Work}\\
			\midrule
			Full
			& \textbf{(adv)+(stoc)}
			& \xmark
			& $\widetilde{\mathcal{O}}(\sqrt{T})$\\
			\addlinespace[3pt]
			& \textbf{(adv)+(adv)}
			&\xmark
			& $\Omega(T)$ \\
			\midrule
            LOB
			& \textbf{(adv)+(stoc)}
			&  $\widetilde{\mathcal{O}}(T^{2/3})$
			& $\widetilde{\mathcal{O}}(\sqrt{T})$ \\
			\bottomrule
	\end{tabular}}
\end{table}

\subsection{Related Works}

Market making has long been studied in market microstructure \citep{bouchaud2008marketsslowlydigestchanges}, with traditional approaches often formulating the choice of bid and ask prices as a stochastic control problem \citep{guéant2017optimalmarketmaking}. An alternative approach treats market making as an online learning problem, following the work of \citet{NIPS2013_995e1fda} and its more recent developments by \citet{CesaBianchi2024MarketMW} and \citet{Maran2026OnlineMM}.

Within this line of research, \citet{CesaBianchi2024MarketMW} study two-bit feedback about private valuation, where the learner observes only whether a purchase and a sale occur, and derive regret guarantees under regularity assumptions on the distributions describing the environment. Although these assumptions support the analysis, they may exclude behavior concentrated around particular prices. For example, the tendency of bidders to bid round numbers documented by \citet{NBERw34564} motivates considering distributions that need not be smooth. 

Alongside the assumptions on the environment, the information available to the learner is another key aspect of the problem. Since two-bit feedback may omit information available in a limit order book, \citet{Maran2026OnlineMM} introduce a richer feedback model, for which they establish an expected regret bound of $\widetilde{\gO}(T^{2/3})$ against an oblivious adversary. Their guarantees are limited to expected regret in the adversarial-price case and do not cover the fully adversarial setting in which both market prices and private valuations are chosen adversarially.

More broadly, market making is related to other online learning problems in economics, such as dynamic pricing \citep{946243.946352}, first-price auctions \citep{doi:10.1137/24M1712308,doi:10.1287/moor.6.1.58} and bilateral trade \citep{pmlr-v195-cesa-bianchi23a,MYERSON1983265}, where the feedback makes it difficult to infer how rewards vary across actions.

\section{Preliminaries}
We study online market making over $T\in\mathbb{N}_{\geq 0}$ rounds. 
At each round $t \in \{1,\dots,T\}$, the learner (\emph{i.e.}, the market maker) posts a bid quote $B_t$ and an ask quote $A_t$, where $A_t \geq B_t$.\footnote{This is a standard no-arbitrage condition. Otherwise, one could buy at the ask and immediately sell at the bid for a guaranteed profit.}
We normalize prices to the interval $[0,1]$, which can always be done through an affine transformation of the original price space. Accordingly, the feasible decision set is:
\[
\gU = \{(a,b) \in [0,1]^2 : a \geq b\}.
\]
At the beginning of round $t \in \{1,\dots,T\}$, a trader with private valuation $V_t \in [0,1]$ arrives.
The trader buys the asset from the market maker whenever $V_t > A_t$, and sells the asset to the market maker whenever $V_t \leq B_t$.
If $V_t \in (B_t,A_t]$, no trade occurs.
We assume that, whenever a trade takes place, the market maker immediately offloads its position on a perfectly liquid secondary market at price $M_t \in [0,1]$, which is observed at the end of the round only, incurring no additional transaction costs or inventory constraints.
The resulting utility function of the market maker is:
\[
u(a,b,v,m) =(a-m)\mathbb{I}\{v>a\}+(m-b)\mathbb{I}\{v\leq b\}.
\]
Throughout the paper, we will use the shorthand notation of $u_t(a,b)=u(a,b,V_t,M_t)$.

The online market making interaction is summarized in Protocol~\ref{alg:protocol_MM_Fullfeed}.
We consider two possible types of feedback: \emph{full} feedback (Line~\ref{full_feedback_protocolline}) and \emph{limit order book} feedback (Lines~\ref{partial_feedback_protocolline4}--\ref{partial_feedback_protocolline7}).
In both cases, the learner observes the market price $M_t$ at the end of each round. The realized trader valuation $V_t$ is always observed at the end of each round under full feedback, while it is observed only if no trade occurs in the limit order book case. 

\begin{protocolalgorithm}[!htp]
\caption{Online Market Making}
\label{alg:protocol_MM_Fullfeed}

\begin{algorithmic}[1]

  \For{$t = 1, \dots , T $}

    \State A new trader with valuation $V_t$ arrives

    \State The learner posts ask and bid prices
    $(A_t, B_t) \in \gU$

    \algstore{mmstate}

\end{algorithmic}

\begin{tcolorbox}[
    enhanced,
    colback=partialteal!4,
    colframe=partialteal,
    boxrule=1.0pt,
    arc=3pt,
    left=0pt,
    right=0pt,
    top=3pt,
    bottom=2pt,
    before skip=3pt,
    after skip=4pt,
    attach boxed title to top right={
        xshift=-0pt,
        yshift=-15pt
    },
    boxed title style={
        colback=white,
        colframe=partialteal,
        boxrule=0.8pt,
        arc=2pt,
        left=4pt,
        right=4pt,
        top=1pt,
        bottom=1pt
    },
    coltitle=partialteal,
    fonttitle={\scriptsize\bfseries},
    title={Limit order book}
]

\begin{algorithmic}[1]

    \algrestore{mmstate}

    \State $M_t \in [0,1]$ is observed
    \label{partial_feedback_protocolline4}

    \If{$B_t < V_t \leq A_t$}

        \State $V_t \in [0,1]$ is observed

    \Else

        \State The learner observes
        $\ind{B_t \geq V_t}$ and
        $\ind{A_t < V_t}$

    \EndIf
    \label{partial_feedback_protocolline7}

    \algstore{mmstate}

\end{algorithmic}

\end{tcolorbox}

\begin{tcolorbox}[
    enhanced,
    colback=fullpurple!4,
    colframe=fullpurple,
    boxrule=1.0pt,
    arc=3pt,
    left=0pt,
    right=0pt,
    top=3pt,
    bottom=2pt,
    before skip=3pt,
    after skip=4pt,
    attach boxed title to top right={
        xshift=-0pt,
        yshift=-15pt
    },
    boxed title style={
        colback=white,
        colframe=fullpurple,
        boxrule=0.8pt,
        arc=2pt,
        left=4pt,
        right=4pt,
        top=1pt,
        bottom=1pt
    },
    coltitle=fullpurple,
    fonttitle=\scriptsize\bfseries,
    title={Full feedback}
]

\begin{algorithmic}[1]

    \algrestore{mmstate}

    \State $(V_t, M_t) \in [0,1]^2$ is observed
    \label{full_feedback_protocolline}

    \algstore{mmstate}

\end{algorithmic}

\end{tcolorbox}

\begin{algorithmic}[1]

    \algrestore{mmstate}

    \State The learner gains utility $u_t(A_t, B_t)$

  \EndFor

\end{algorithmic}

\end{protocolalgorithm}

We consider two different models for the sequence $\{(V_t,M_t)\}_{t=1}^T$: 

\begin{itemize}[leftmargin=65pt, labelsep=0.6em]
    \item[\textbf{(adv) + (stoc)}]
    $\{V_t\}_{t=1}^T$ is an i.i.d.\ sequence, while $\{M_t\}_{t=1}^T$ is chosen by an adaptive adversary.

    \item[\textbf{(adv) + (adv)}]
    Fully adversarial setting, in which the sequence $\{(V_t,M_t)\}_{t=1}^T$ is generated by an oblivious adversary.
\end{itemize}

The first case \textbf{(adv) + (stoc)} is the main focus of this work, while we provide a lower bound, which naturally extends to the case of an adaptive adversary, for the second case \textbf{(adv) + (adv)}.

Let us remark that we make no assumption on the regularity of the sequence of $\{V_t\}_{t=1}^T$. Indeed, differently from~\citet{CesaBianchi2024MarketMW}, we do \emph{not} assume that the c.d.f.\ of the trader valuations distribution is Lipschitz--continuous in any part of the work. 

\subsection{Performance Metric} 
We use as performance metric the notion of \emph{cumulative regret}, which is defined as follows: 
\[
R_T = \sup_{(a,b)\in \gU} \mathbb E_V \left[ \sum_{t=1}^T  \left( u_t(a,b)- u_t(A_t,B_t) \right)\right] ,
\]
where the expectation is taken over the randomization of the environment. 
We define the c.d.f.\ for the random variable $V$ underlying the trader valuations in the \textbf{(adv)+(stoc)} as $F(x)=\mathbb P (V\leq x)$ and, accordingly, we define the expected utility as: \[U_t(a,b) \coloneqq \mathbb E_V [u_t(a,b)] =(M_t-b)F(b) + (a-M_t)(1-F(a)).\]
We notice that the utility can be decomposed into two components, one for the bid and one for the ask price. 
Thus, we denote: \[u_t^A(a)\coloneqq(a-M_t)\ind{V_t > a}, \quad \quad \quad  u_t^B(b)\coloneqq(M_t-b)\ind{V_t \leq b},\] so that $u_t(a,b) = u_t^B(b) + u_t^A(a)$.
Defining $\mu_T = \frac{1}{T}\sum_{t=1}^T M_t$, we also notice that: \[\sup_{(a,b) \in \gU} \frac{1}{T} \sum_{t=1}^T U_t(a,b) =\sup_{(a,b) \in \gU} \bigg[ (\mu_T - b)F(b) + (a -\mu_T)(1-F(a))\bigg].\]
Finally, defining $U(a,b)\coloneqq (\mu_T - b)F(b) + (a -\mu_T)(1-F(a))$, we can rewrite the regret as: 
\begin{align}\label{eq:meanregret_def}
    R_T = \sup_{(a,b) \in \gU} \sum_{t=1}^T \bigg[U_t(a,b) - U_t(A_t,B_t) \bigg]= T \cdot \sup_{(a,b) \in \gU}  U(a,b) - \sum_{t=1}^T U_t(A_t, B_t).
\end{align}
Notice that, in case of adversarial trader valuations, where no statistical assumption is made on the sequence $\{(M_t, V_t)\}_{t=1}^T$, no expectation is taken in the regret definition.
 
In the next section, we start with assuming full feedback and stochastic trader valuations. 

\section{Warm up: The Full Feedback Case}\label{sec:advstoc_fullfeed}

In this section, we start by considering the \textbf{(adv)+(stoc)} scenario.

This setting is the one already studied in \citep{Maran2026OnlineMM}, with the difference that we assume, throughout this section, to have full feedback, while we let the adversary to be adaptive.
Our approach will employ two instances of the Hedge algorithm~\citep{FREUND1997119}, thus significantly departing from the existing technique of~\citet{Maran2026OnlineMM}, where the algorithm  is mainly based on Follow the Perturbed Leader (FTPL) \citep{JMLR:v6:hutter05a}.

\subsection{Algorithm}

\begin{algorithm}[!htp]
\caption{Market Making with Experts}
\label{alg:Hedge_onthegrid}  
    \begin{algorithmic}[1]
     \Input{Time horizon $T$, $K$, $\eta$}   
        \State Initialize $A_1 \gets 1$, $B_1 \gets 0$, and grid $\gG$
        \State Initialize $W_1^A, W_1^B$ to 1 on $\gG$ \label{Hedge_onthegrid_line2} 
        \For{$t = 1, \dots, T$}
            \State Form $w_t^A,w_t^B$ by normalizing $W_t^A, W_t^B$ 
            \State Draw $Y_t \sim w^A_t$ and  $X_t \sim w^B_t$
            \State Post prices $(A_t,B_t) \gets(\max\{ X_t,Y_t\}, \min\{ X_t,Y_t\})$ \label{Hedge_onthegrid_line6}
            \State Gain utility $u_t(A_t,B_t)$
            \State Observe $(M_t,V_t)$ and reconstruct $u_t^A(a), u_t^B(b)$ for all $a,b \in \gG$ \label{Hedge_onthegrid_line7}
            \State Update weights for all $a,b \in \gG$: \[W_{t+1}^A(a) \gets W^A_t(a)\exp\left(\eta\,\, u_t^A(a)\right),\quad W_{t+1}^B(b) \gets W^B_t(b)\exp\left(\eta\,\, u_t^B(b)\right), \]\label{Hedge_onthegrid_line8}
        \EndFor
    \end{algorithmic}
\end{algorithm}

Our algorithm, whose pseudocode is provided in Algorithm~\ref{alg:Hedge_onthegrid}, discretizes the bid and ask prices using the uniform grid $
\gG = \left\{\frac{i}{K} : i = 0,\dots,K\right\}.
$
The resulting Cartesian product $\gG \times \gG$ consists of $(K+1)^2$ points in $[0,1]^2$ and serves as the discrete action set.
We then run two independent instances of Hedge~\citep{Orabona} with the same learning rate $\eta>0$, one for bid prices and one for ask prices. Each instance maintains weights $\{W_t^B\}_{t=1}^T$ and $\{W_t^A\}_{t=1}^T$, respectively, and the corresponding probability distributions $\{w_t^B\}_{t=1}^T$ and $\{w_t^A\}_{t=1}^T$ over the one-dimensional grid $\gG$ (Line~\ref{Hedge_onthegrid_line2}).
At each round, the algorithm samples a bid and an ask price from the corresponding distributions and swaps them if the resulting pair violates the constraint defining $\gU$ (Line~\ref{Hedge_onthegrid_line6}). It then posts the resulting pair of prices and updates the weights of both Hedge instances (Line~\ref{Hedge_onthegrid_line8}). We emphasize that the swap only affects the prices played, not the weight updates. Since feedback is fully observable, at the end of each round we can evaluate the utility associated with every grid price in its respective role as a bid or an ask. Each Hedge instance is therefore updated using these utilities, regardless of whether the sampled prices were swapped (Lines~\ref{Hedge_onthegrid_line7}--\ref{Hedge_onthegrid_line8}).

\subsection{Theoretical Results}

In this section, we analyze Algorithm~\ref{alg:Hedge_onthegrid} and establish the associated $\tilde{\gO}(\sqrt{T})$ regret bound.

We start by stating the following lemma, which shows how enforcing feasibility by swapping the prices predicted by the two Hedge instances never decreases the utility. More precisely, the utility remains unchanged when the sampled prices already satisfy the constraint, and increases by $X_t-Y_t$ otherwise. Thus, the swap in Line~\ref{Hedge_onthegrid_line6} of Algorithm~\ref{alg:Hedge_onthegrid} guarantees a pair of prices whose utility is at least as large as that of the original predictions. This property will allow us to transfer the utility guarantees of the two Hedge instances to the actions actually played by the algorithm.\footnote{When full feedback is available, this lemma is not necessary. Indeed, it is possible to run a single Hedge instance on the two-dimensional grid to recover the same regret guarantees. We employ two Hedge instances invoking Lemma~\ref{lemma:order_the_game} as a warm-up to the next section.}

\begin{restatable}{lemma}{lemmaOderTheGame}
    \label{lemma:order_the_game}
        Consider $X_t\sim w^B_t$ and $Y_t \sim w^A_t$. Then it holds that:
    \[
    u_t(\max\{ X_t,Y_t\}, \min\{ X_t,Y_t\}) = u_t(Y_t,X_t) + (X_t - Y_t)_+,
    \]
    where we denote $(x)_+=\max\{0,x \}$. 
    As a consequence, for any $t \in \{1, \dots, T\}$ it holds that $U_t(A_t,B_t) \geq U_t(Y_t,X_t)$ and $U(A_t,B_t) \geq U(Y_t,X_t)$.
\end{restatable}
Next, we show that restricting the action set to the grid incurs only a bounded loss in expected utility. Specifically, there exists a grid point whose expected utility is within $\nicefrac{5}{2K}$ of the supremum over the continuous feasible set $\gU$. This allows us to relate performance on the grid to the continuous optimum.
\begin{restatable}{lemma}{preregretOMDreduction}
\label{lemma:pre_regretOMDreduction}
        There exists a point $(A^*,B^*) \in  \gG \times \gG$ such that:
        $$
        0 \leq \sup_{(a,b)\in \gU} U(a,b)-U(A^*,B^*) \leq \frac{5}{2K}.
        $$
\end{restatable}
\begin{figure}[t]
\centering
\resizebox{0.85\linewidth}{!}{%
\begingroup

\definecolor{omdInk}{HTML}{243248}
\definecolor{omdMuted}{HTML}{657286}
\definecolor{omdGrid}{HTML}{CBD2DC}
\definecolor{omdGridDot}{HTML}{ABB5C2}
\definecolor{omdBlue}{HTML}{276AC4}
\definecolor{omdRed}{HTML}{D6524D}
\definecolor{omdBlueWash}{HTML}{E8F1FC}
\definecolor{omdRoseWash}{HTML}{FAF0F0}
\definecolor{omdHeader}{HTML}{F1F4FA}

\begin{tikzpicture}[
  x=1cm,y=1cm,
  font=\sffamily\fontsize{10.5}{13}\selectfont,
  text=omdInk,
  axis/.style={
    -{Stealth[length=2.2mm,width=1.5mm]},
    draw=omdInk,line width=.85pt
  },
  threshold/.style={
    draw=omdMuted,line width=.85pt,
    dash pattern=on 4pt off 3pt
  },
  rounding/.style={
    -{Stealth[length=2.7mm,width=2mm]},
    draw=omdBlue,line width=1.8pt,
    rounded corners=2pt,
    shorten <=5pt,shorten >=5pt
  },
  pointlabel/.style={
    font=\fontsize{12}{14}\selectfont,
    fill=white,fill opacity=.93,text opacity=1,
    inner sep=2pt
  },
  dimension/.style={
    draw=omdMuted,line width=.65pt,
    decorate,
    decoration={brace,mirror,amplitude=3pt}
  },
  paneltitle/.style={
    font=\sffamily\bfseries\fontsize{15}{18}\selectfont
  },
  panelsubtitle/.style={
    font=\sffamily\fontsize{10}{12}\selectfont,
    text=omdMuted
  },
  explanation/.style={
    font=\sffamily\fontsize{10.5}{13}\selectfont,
    text=omdMuted,align=center,text width=7cm
  }
]

\begin{scope}

  \fill[omdHeader,rounded corners=5pt]
    (-.56,6.72) rectangle (6.56,7.72);
  \node[paneltitle] at (3,7.36)
    {Inward rounding};
  \node[panelsubtitle] at (3,6.96)
    {The threshold is not crossed};

  \fill[omdRoseWash] (0,0) rectangle (6,6);
  \fill[omdBlueWash] (2.65,0) rectangle (6,2.65);
  \fill[omdBlue!14] (4,1) rectangle (5,2);

  \draw[step=1,omdGrid,line width=.45pt]
    (0,0) grid (6,6);
  \foreach \i in {0,...,6} {
    \foreach \j in {0,...,6} {
      \fill[omdGridDot]
        (\i,\j) circle[radius=1.15pt];
    }
  }

  \draw[omdBlue!58,line width=.85pt]
    (4,1) rectangle (5,2);

  \draw[threshold] (2.65,0) -- (2.65,6.10);
  \draw[threshold] (0,2.65) -- (6.10,2.65);

  \draw[axis] (-.16,0) -- (6.35,0)
    node[right=2pt] {$a$};
  \draw[axis] (0,-.16) -- (0,6.35)
    node[above=2pt] {$b$};

  \node[below=6pt,font=\fontsize{12}{14}\selectfont]
    at (2.65,0) {$\mu_T$};
  \node[left=6pt,font=\fontsize{12}{14}\selectfont]
    at (0,2.65) {$\mu_T$};

  \coordinate (P) at (4.88,1.13);
  \coordinate (Q) at (4,2);

  \draw[rounding] (P) -- (4,1.13) -- (Q);

  \filldraw[fill=omdRed,draw=white,line width=1pt]
    (P) circle[radius=4pt];
  \filldraw[fill=omdBlue,draw=white,line width=1pt]
    (Q) circle[radius=4pt];

  \node[pointlabel,text=omdRed,anchor=north]
    at (5.03,.78) {$(a^*,b^*)$};
  \draw[omdRed!60,line width=.6pt]
    (4.99,.82) -- (4.91,1.02);

  \node[pointlabel,text=omdBlue,anchor=south west]
    at (4.25,2.03) {$(A^*,B^*)$};
  \draw[omdBlue!65,line width=.6pt]
    (4.24,2.10) -- (4.06,2.03);

  \draw[dimension] (4,-.86) -- (5,-.86)
    node[midway,below=7pt,
         font=\fontsize{11}{13}\selectfont] {$1/K$};

  \node[explanation] at (3,-1.83)
    {Move left and up to a grid node.};

\end{scope}

\begin{scope}[xshift=8.55cm]

  \fill[omdHeader,rounded corners=5pt]
    (-.56,6.72) rectangle (6.56,7.72);
  \node[paneltitle] at (3,7.36)
    {Boundary correction};
  \node[panelsubtitle] at (3,6.96)
    {Close-up of the threshold cell};

  \fill[omdRoseWash] (0,0) rectangle (6,6);
  \fill[omdBlueWash] (3.30,0) rectangle (6,3.30);

  \draw[step=2,omdGrid,line width=.45pt]
    (0,0) grid (6,6);
  \foreach \i in {0,2,4,6} {
    \foreach \j in {0,2,4,6} {
      \fill[omdGridDot]
        (\i,\j) circle[radius=1.15pt];
    }
  }

  \draw[omdBlue!58,line width=.85pt]
    (2,2) rectangle (4,4);

  \draw[threshold] (3.30,0) -- (3.30,6.10);
  \draw[threshold] (0,3.30) -- (6.10,3.30);

  \draw[axis] (-.16,0) -- (6.35,0)
    node[right=2pt] {$a$};
  \draw[axis] (0,-.16) -- (0,6.35)
    node[above=2pt] {$b$};

  \node[below=6pt,font=\fontsize{12}{14}\selectfont]
    at (3.30,0) {$\mu_T$};
  \node[left=6pt,font=\fontsize{12}{14}\selectfont]
    at (0,3.30) {$\mu_T$};

  \coordinate (Pb) at (3.50,3.00);
  \coordinate (Qb) at (4,2);
  \coordinate (Rejected) at (2,4);

  \draw[
    -{Stealth[length=2.1mm,width=1.4mm]},
    omdRed!65,line width=1pt,
    dash pattern=on 3.5pt off 3pt,
    shorten <=5pt,shorten >=7pt
  ] (Pb) -- (Rejected);

  \draw[rounding] (Pb) -- (4,3) -- (Qb);

  \filldraw[fill=omdRed,draw=white,line width=1pt]
    (Pb) circle[radius=4pt];
  \filldraw[fill=omdBlue,draw=white,line width=1pt]
    (Qb) circle[radius=4pt];

  \draw[white,line width=4.6pt]
    (1.89,3.89) -- (2.11,4.11)
    (1.89,4.11) -- (2.11,3.89);
  \draw[omdRed!80,line width=1.45pt,line cap=round]
    (1.89,3.89) -- (2.11,4.11)
    (1.89,4.11) -- (2.11,3.89);

  \node[
    align=center,text=omdRed!85!omdInk,
    font=\sffamily\fontsize{10}{12}\selectfont
  ] at (1.60,4.72)
    {Discarded:\\crosses the threshold};

  \draw[omdRed!50,line width=.6pt]
    (1.84,4.32) -- (1.96,4.14);

  \node[pointlabel,text=omdRed,anchor=west]
    at (4.18,3.50) {$(a^*,b^*)$};
  \draw[omdRed!60,line width=.6pt]
    (4.17,3.47) -- (3.60,3.08);

  \node[pointlabel,text=omdBlue,anchor=west]
    at (4.20,1.65) {$(A^*,B^*)$};
  \draw[omdBlue!65,line width=.6pt]
    (4.22,1.76) -- (4.06,1.94);

  \draw[dimension] (2,-.86) -- (4,-.86)
    node[midway,below=4pt,
         font=\fontsize{11}{13}\selectfont] {$1/K$};

  \node[explanation] at (3,-1.90)
    {Here, the safe grid node lies right and down.};

\end{scope}

\draw[omdGrid!65,line width=.65pt]
  (-.56,-2.34) -- (15.11,-2.34);

\filldraw[fill=omdRed,draw=white,line width=.7pt]
  (.02,-2.89) circle[radius=3.4pt];
\node[anchor=west] at (.25,-2.89)
  {Optimal action};

\filldraw[fill=omdBlue,draw=white,line width=.7pt]
  (5.03,-2.89) circle[radius=3.4pt];
\node[anchor=west] at (5.26,-2.89)
  {Discretized action};

\filldraw[
  fill=omdBlueWash,draw=omdBlue!35,line width=.55pt
] (9.78,-3.05) rectangle (10.10,-2.73);
\node[anchor=west] at (10.25,-2.89)
  {$b\leq\mu_T\leq a$};

\end{tikzpicture}
\endgroup
}
\caption{Discretization around the optimal action.
\emph{Left:} inward rounding.
\emph{Right:} correction when both coordinates would cross $\mu_T$.}
\label{fig:grid-discretization}
\end{figure}
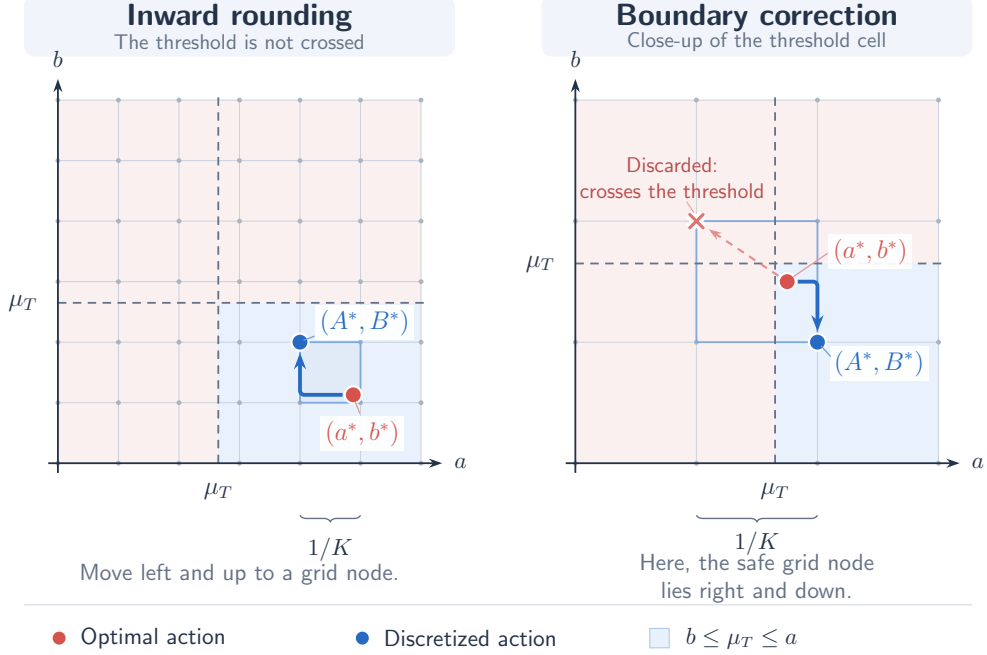
Figure~\ref{fig:grid-discretization} illustrates the idea behind the proof of Lemma~\ref{lemma:pre_regretOMDreduction}. We start from a nearly optimal point and choose a nearby point on the grid, taking care not to cross the boundary defined by $\mu_T$, as doing so could reduce the expected utility. This choice keeps the discretization loss small. We provide the full argument in Appendix~\ref{app:advstoc_fullfeed}.

We are now ready to provide the final regret bound attained by Algorithm~\ref{alg:Hedge_onthegrid}.

\begin{restatable}{theorem}{Thregretfull}
    \label{th:regret_fullfeedadvstoc}
    Let $\delta\in (0,1)$ and $T \geq 2$.
    Algorithm~\ref{alg:Hedge_onthegrid} with $K= \lceil \sqrt{T} \rceil, \eta = \sqrt \frac{\log(T)}{T}$ attains regret: $$R_T \leq  \gO \left(\sqrt{T \log(T)} + \sqrt{T \log(1/\delta)}+ \log(1/\delta) \right),$$
    with probability at least $1-\delta$.
\end{restatable}
In Theorem~\ref{th:regret_fullfeedadvstoc}, we derive a regret bound by adding the standard Hedge guarantees for the two one-dimensional instances. We analyze each instance using its original sampled quote, before any swap takes place. Although the swap couples the bid and ask prices actually posted, Lemma~\ref{lemma:order_the_game} shows that it never decreases utility.
We can therefore analyze the two Hedge instances separately and carry their combined guarantee over to the actions played by the algorithm.
The regret bound in Theorem~\ref{th:regret_fullfeedadvstoc} depends only logarithmically on the grid size $K$, as a consequence of the Hedge guarantee. This dependence is optimal in the standard expert setting. We would expect a different dependence on $K$ under limit order book feedback, which provides less information than full feedback.

\section{Optimal Bounds For The Limit Order Book Case}\label{sec:advstoc_LOB}
In this section, we study the more interesting setting of the limit order book feedback, where the information available to the learner depends on whether a trade occurs. 
If no exchange takes place, the learner receives full feedback; otherwise, only the utility is observed with the type of trade.

\subsection{Algorithm}

\begin{algorithm}[t]
  \caption{Market Making with Limit Order Book}
\label{alg:MM_partialfeed_exp3}  
\begin{algorithmic}[1]
     \Input{Time horizon $T$, parameters $\eta, \beta, K$}   
     \State Initialize $\mathcal{G}_A,\mathcal{G}_B$ \label{Partial_feedback_line1}
     \State Set $W_1^A,W_1^B$ to 1 \label{Partial_feedback_line2}
     \For{$t = 1, \dots , T$}
        \State Form probabilities $w_t^A,w_t^B$ normalizing $W_t^A,W_t^B$\label{Partial_feedback_line4}
        \State Compute $Q_t$ with Equation~\eqref{eq:Qderivation} in internal grid points, $Q_t(0)=Q_t(1)=1$
        \State Draw $X_t \sim w^B_t$ and $Y_t \sim w^A_t$\label{Partial_feedback_line6}
        \State Play prices $(A_t,B_t)=(\max\{X_t,Y_t\},\min\{X_t,Y_t\})$
        \State Observe $M_t$, observe $V_t$ if no trade occurred
        \State $\forall z \in (\mathcal{G}_A \cup \mathcal{G}_B) \cap [B_t,A_t] $, compute $ \ell^s_t(z)$, for $s \in \{ A,B\}$ respectively
        \State Form the loss estimator $\hat \ell^s_t$, with $s \in \{ A,B\}$, as in Equation~\eqref{def:ellcap_int}
        \State Update weights $W_{t+1}^A,W_{t+1}^B$
     \EndFor
\end{algorithmic}
\end{algorithm}

We now present Algorithm~\ref{alg:MM_partialfeed_exp3} for the limit order book feedback and analyze its regret.

Similarly to the previous approach, Algorithm~\ref{alg:MM_partialfeed_exp3} treats the bid and ask prices separately, maintaining disjoint grids $\gG_B$ and $\gG_A$ (Line~\ref{Partial_feedback_line1}), weights $W_t^B$ and $W_t^A$ (Line~\ref{Partial_feedback_line2}), and sampling distributions $w^A_t$ and $w^B_t$ (Lines~\ref{Partial_feedback_line4}--\ref{Partial_feedback_line6}).
Unlike in the full-feedback setting, we must account for the possibility that grid points coincide with atoms of $F$\footnote{We assume throughout the paper that $F$ has no atom at zero. This reflects the fact that traders in financial markets typically assign strictly positive valuations to the assets they trade.}, which might affect the observed feedback. 
This issue is addressed in Lemma~\ref{lemma:partfeed_reconstruction} in Appendix~\ref{app:advstoc_LOB}, where we show that, independently of the underlying distribution of $\{V_t\}_{t=1}^T$, we can still collect the necessary information for our estimation.

For an integer $K \geq 1$, we define the two grids as:
\begin{align*}
    \gG_B = \left\{\frac{i}{K} + \frac{1}{4K}: i=0,\dots,K-1\right\} \cup \{1\}, \
    \gG_A = \left\{ \frac{i}{K} + \frac{1}{2K}: i=0,\dots,K-1\right\} \cup \{0\}.
\end{align*}
Each grid contains $K$ internal points and one endpoint.
We run two parallel instances of an EXP3-based algorithm~\citep{doi:10.1137/S0097539701398375}. For convenience, we express the utility contributions as losses in $[0,1]$ by defining
$
    \ell_t^A(a) := \nicefrac{(1-u_t^A(a))}{2},
    \ell_t^B(b) := \nicefrac{(1-u_t^B(b))}{2}.
$
Let $\beta>0$ be an implicit exploration parameter \citep{neu}.
For each internal grid point $x \in \gG_s$, with $s \in \{A,B\}$, let
$
    Q_t(x) := \mathbb{P}\bigl(x \in [B_t,A_t] \cup \{0,1\}\bigr)
$
denote the probability that its loss is observed at round $t$. 
We set $Q_t(1)=Q_t(0)=1$, since losses at the endpoints can always be computed from $M_t$, which is observed at every round.
We further denote $O_t(x) \coloneqq \ind{x \in [B_t,A_t] \cup \{0,1\}}$ as the event that a point $x \in \gG_s, \,\, s \in \{A,B\}$ has deducible utility. 
The limit order book feedback allows us to recover the exact losses at all grid points in $[B_t,A_t]$, as we will show in Lemma~\ref{lemma:partfeed_reconstruction}. We use these observed losses to construct the estimates:
\begin{align}\label{def:ellcap_int}
    \hat \ell_t^A(x) &= \frac{O_t(x) \ell_t^A(x)}{Q_t(x)+\beta},
    \qquad
    \hat \ell_t^B(x) = \frac{O_t(x) \ell_t^B(x)}{Q_t(x)+\beta}.
\end{align}
The implicit exploration parameter $\beta$ prevents the estimates from becoming too large when a grid point is rarely observed, helping to control their variance. 

For each component $s \in \{A,B\}$, we update the weights according to:
\[
    W_{t+1}^s(x) = W_t^s(x)\exp\bigl(-\eta\hat\ell_t^s(x)\bigr),
\]
and denote the corresponding normalized weights by $w_t^s, \, s\in \{A,B\}$, which represent the probability of sampling for the two problems (Line~\ref{Partial_feedback_line6}). 

Finally, we notice that for any internal grid point $x \in \gG_s$, with $s \in \{A,B\}$, we can compute the observation probability directly from the sampling distributions:
\begin{align}\label{eq:Qderivation}
    Q_t(x) = 1-\left(\sum_{y \in \gG_A: y< x}w^A_t(y) \right) \hspace{-0.2cm} \left(\sum_{y \in \gG_B: y< x}w^B_t(y)\right)\hspace{-0.1cm}- \hspace{-0.1cm} \left(\sum_{y \in \gG_A: y> x}w^A_t(y) \right)\hspace{-0.2cm}\left(\sum_{y \in \gG_B:  y> x}w^B_t(y)\right).
\end{align}
Indeed, $x$ lies outside the played interval $[B_t,A_t]$ precisely when both sampled prices are strictly below $x$ or both are strictly above it. Since the prices are sampled independently, the probabilities of these two events are given by the products above. 

\subsection{Theoretical results}
In this section, we establish the regret bound for Algorithm~\ref{alg:MM_partialfeed_exp3}. 

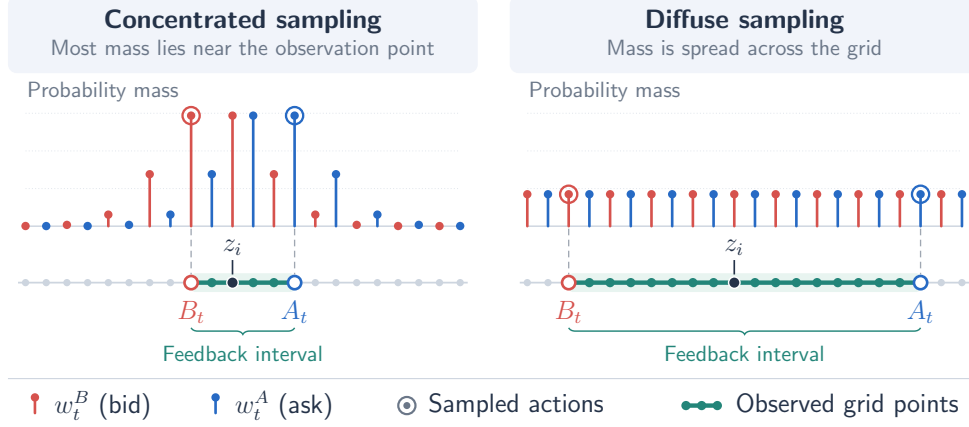
\begin{figure}[t]
\centering
\resizebox{0.85\linewidth}{!}{%
\begingroup
\definecolor{fbInk}{HTML}{243248}
\definecolor{fbMuted}{HTML}{6B7788}
\definecolor{fbGrid}{HTML}{CDD5DF}
\definecolor{fbBlue}{HTML}{276AC4}
\definecolor{fbRed}{HTML}{D6524D}
\definecolor{fbGreen}{HTML}{238579}
\definecolor{fbWash}{HTML}{E8F4F0}
\definecolor{fbHeader}{HTML}{F1F4FA}

\def\fbConcentratedP{{.00051,.00371,.03288,.14796,.31477,.31477,
                     .14796,.03288,.00371,.00051,.00034}}
\def\fbConcentratedQ{{.00034,.00051,.00371,.03288,.14796,.31477,
                     .31477,.14796,.03288,.00371,.00051}}
\def\fbDiffuse{{.09091,.09091,.09091,.09091,.09091,.09091,
               .09091,.09091,.09091,.09091,.09090}}

\begin{tikzpicture}[
  x=1cm,y=1cm,
  font=\sffamily\fontsize{12}{14}\selectfont,text=fbInk,
  mass/.style={line width=1.2pt,line cap=round},
  guide/.style={draw=fbMuted!65,densely dashed,line width=.6pt},
  title/.style={font=\sffamily\bfseries\fontsize{12}{16}\selectfont},
  subtitle/.style={font=\sffamily\fontsize{10}{12}\selectfont,
                   text=fbMuted},
  mathlabel/.style={font=\fontsize{12}{14}\selectfont,inner sep=2pt}
]

\newcommand{\fbPanel}[7]{%
\begin{scope}[xshift=#1cm]
  \fill[fbHeader,rounded corners=4pt]
    (-.28,3.34) rectangle (7.21,4.57);
  \node[title] at (3.465,4.17) {#2};
  \node[subtitle] at (3.465,3.72) {#3};

  \node[subtitle,anchor=west] at (-.10,3.04)
    {Probability mass};
  \foreach \y in {1.50,2.10,2.70}
    \draw[fbGrid!60,densely dotted,line width=.45pt]
      (0,\y) -- (6.93,\y);
  \draw[fbGrid,line width=.65pt] (0,.90) -- (6.93,.90);

  \foreach \j in {0,...,10}{
    \pgfmathsetmacro{\xp}{.33*(2*\j)}
    \pgfmathsetmacro{\xq}{.33*(2*\j+1)}
    \pgfmathsetmacro{\hp}{.90+5.6*#4[\j]}
    \pgfmathsetmacro{\hq}{.90+5.6*#5[\j]}
    \draw[mass,fbRed] (\xp,.90) -- (\xp,\hp);
    \fill[fbRed] (\xp,\hp) circle[radius=1.9pt];
    \draw[mass,fbBlue] (\xq,.90) -- (\xq,\hq);
    \fill[fbBlue] (\xq,\hq) circle[radius=1.9pt];
  }

  \pgfmathsetmacro{\xb}{.33*#6}
  \pgfmathsetmacro{\xa}{.33*#7}
  \pgfmathtruncatemacro{\jb}{#6/2}
  \pgfmathtruncatemacro{\ja}{(#7-1)/2}
  \pgfmathsetmacro{\hb}{.90+5.6*#4[\jb]}
  \pgfmathsetmacro{\ha}{.90+5.6*#5[\ja]}

  \draw[fbRed,line width=1.05pt]
    (\xb,\hb) circle[radius=3.9pt];
  \draw[fbBlue,line width=1.05pt]
    (\xa,\ha) circle[radius=3.9pt];
  \draw[guide] (\xb,.86) -- (\xb,.13);
  \draw[guide] (\xa,.86) -- (\xa,.13);

  \fill[fbWash,rounded corners=2pt]
    (\xb,-.15) rectangle (\xa,.15);
  \draw[fbGrid,line width=.8pt] (-.12,0) -- (7.05,0);
  \draw[fbGreen,line width=2pt] (\xb,0) -- (\xa,0);
  \foreach \i in {0,...,21}{
    \pgfmathsetmacro{\xx}{.33*\i}
    \fill[fbGrid] (\xx,0) circle[radius=1.65pt];
  }
  \foreach \i in {#6,...,#7}{
    \pgfmathsetmacro{\xx}{.33*\i}
    \fill[fbGreen] (\xx,0) circle[radius=2.1pt];
  }
  \filldraw[fill=white,draw=fbRed,line width=1.2pt]
    (\xb,0) circle[radius=3pt];
  \filldraw[fill=white,draw=fbBlue,line width=1.2pt]
    (\xa,0) circle[radius=3pt];
  \node[mathlabel,text=fbRed,below=6pt] at (\xb,0) {$B_t$};
  \node[mathlabel,text=fbBlue,below=6pt] at (\xa,0) {$A_t$};

  \draw[fbInk,line width=.75pt] (3.30,.13) -- (3.30,.36);
  \filldraw[fill=fbInk,draw=white,line width=.65pt]
    (3.30,0) circle[radius=2.6pt];
  \node[mathlabel,above=0pt] at (3.30,.36) {$z_i$};

  \draw[fbGreen,line width=.65pt,decorate,
    decoration={brace,mirror,amplitude=3pt}]
    (\xb,-.72) -- (\xa,-.72);
  \node[text=fbGreen,font=\sffamily\fontsize{10}{12}\selectfont]
    at ({(\xb+\xa)/2},-1.13) {Feedback interval};
\end{scope}%
}

\fbPanel{0}{Concentrated sampling}
  {Most mass lies near the observation point}
  {\fbConcentratedP}{\fbConcentratedQ}{8}{13}
\fbPanel{8}{Diffuse sampling}
  {Mass is spread across the grid}
  {\fbDiffuse}{\fbDiffuse}{2}{19}

\draw[fbGrid,line width=.5pt] (-.28,-1.53) -- (15.21,-1.53);
\draw[mass,fbRed] (.15,-2.08) -- (.15,-1.83);
\fill[fbRed] (.15,-1.83) circle[radius=1.9pt];
\node[anchor=west] at (.34,-1.96) {$w^B_t$ (bid)};
\draw[mass,fbBlue] (3.03,-2.08) -- (3.03,-1.83);
\fill[fbBlue] (3.03,-1.83) circle[radius=1.9pt];
\node[anchor=west] at (3.22,-1.96) {$w^A_t$ (ask)};
\draw[fbMuted,line width=1pt]
  (6.08,-1.96) circle[radius=3.9pt];
\fill[fbMuted] (6.08,-1.96) circle[radius=1.6pt];
\node[anchor=west] at (6.30,-1.96) {Sampled actions};
\draw[fbGreen,line width=2pt] (10.53,-1.96) -- (11.03,-1.96);
\foreach \x in {10.53,10.78,11.03}
  \fill[fbGreen] (\x,-1.96) circle[radius=2.1pt];
\node[anchor=west] at (11.20,-1.96) {Observed grid points};

\end{tikzpicture}
\endgroup%
}
\caption{{\small 
The two panels illustrate sampling under a narrower (left) and wider (right) bid--ask spread. The bars $w^B_t$ and $w^A_t$ denote the sampling probabilities over prices, while $Q_t(z_i)$ is the probability of observing an internal point $z_i\in[B_t,A_t]$. 
Concentrating the sampling mass induces greater utility, but weakens observability of internal points. Consequently, although $Q_t(z_i)$ may be small, the aggregate importance-weight ratios $w^B_t(z_i)/Q_t(z_i)$ and $w^A_t(z_i)/Q_t(z_i)$ remain bounded, yielding the $\mathcal{O}(\log K)$ dependence in Lemma~\ref{lemma:bound_logweights}, provided that $\beta \propto 1/K$.
}}
\label{fig:feedback-mass}
\end{figure}

We begin by showing that the available feedback is sufficient to compute the loss estimates used by the algorithm, and then use these estimates to analyze its performance. 
Specifically, the first lemma shows that, despite the limited feedback, playing a price pair $(A_t,B_t)$ allows us to recover the exact losses at every grid point between the two quotes.
Thus, a single round provides information about an entire interval of grid points, which we use to construct the estimates in Equation~\eqref{def:ellcap_int}.
\begin{restatable}{lemma}{partfeedreconstruction}
\label{lemma:partfeed_reconstruction}
    Let $(A_t,B_t)$ be a price pair with $B_t \leq A_t$, played at time step $t$.
    Then, for any point $x\in (\gG_A \cup \gG_B)  \cap [B_t,A_t]$, it is possible to reconstruct $\ell^s_t(x), \,\, s \in \{A,B\}$, after playing $(A_t,B_t)$.
\end{restatable}
Lemma~\ref{lemma:partfeed_reconstruction} follows from the structure of the limit order book feedback. If no trade occurs, the learner receives full feedback. Otherwise, observing whether a purchase or a sale occurred reveals on which side of the bid--ask spread the private value lies. In either case, the learner can determine how the private value compares with each grid point within the spread and, since $M_t$ is always observed, reconstruct the corresponding utility.

Next, Lemma~\ref{lemma:bound_logweights} bounds the ratios between the probability of sampling each grid point and the probability of observing its loss.
These ratios measure how much information the feedback provides relative to how often each point is selected. Since each round reveals losses throughout the interval between the sampled quotes, the sum of these ratios grows only logarithmically with the grid size.
\begin{restatable}{lemma}{boundlogweights}
\label{lemma:bound_logweights}
    For any time step $t \in \{1, \dots, T \}$, with $T \geq 2$, it holds:
    \begin{align}\label{eq:bound_logweights}
    \sum_{x \in  \mathcal{G}_s ,\, s\in \{ A,B\}} \frac{w_t^s(x)}{Q_t(x) + \beta}  \leq 4 \left( \log \left(\frac{1 + \beta}{\beta}\right)\right).
\end{align}
\end{restatable}

Figure~\ref{fig:feedback-mass} illustrates the result of Lemma~\ref{lemma:bound_logweights}. Specifically, reconstructing the losses over the interval $[B_t,A_t]$ allows us to bound the term measuring the cost of the limit order book feedback by $\gO(\log K)$. This lies between the full-feedback setting, where every loss is observed and the corresponding term is $\gO(1)$, and the bandit setting, where only the loss of the sampled action is observed and the term scales as $\gO(K)$.

We are now ready to show the regret guarantees of Algorithm~\ref{alg:MM_partialfeed_exp3}. 

\begin{restatable}{theorem}{partialregrettheorem}
\label{th:partialregrettheorem}
    Let $\delta \in (0,1)$ and $T \geq 2$. By selecting $\eta=\beta=\sqrt{1/T}$, Algorithm~\ref{alg:MM_partialfeed_exp3} instantiated with $K=\lceil \sqrt{T}\rceil$ attains:
    \[
    R_T =  \gO \left(\sqrt{T} \log\left( \frac{T}{\delta}\right)\right),
    \]
    with probability at least $1-\delta$.
\end{restatable}
Theorem~\ref{th:partialregrettheorem} improves upon the bound of \citet{Maran2026OnlineMM} by establishing a high-probability regret bound of $\widetilde{\gO}(\sqrt{T})$, which is optimal up to logarithmic factors and holds even against an adaptive adversary.

\section{An impossibility result}\label{sec:LB_fullfeedMM}
We finally consider the full-feedback setting under the \textbf{(adv)+(adv)} case. We establish a linear regret lower bound, showing that market making is unlearnable in a fully adversarial environment, even when the entire valuations are observed at the end of each round.

This is done in the following theorem.
\begin{restatable}{theorem}{LBforMM}
    \label{th:adv_LB}
        Consider the market making problem with adversarial sequence of $\{M_t\}_t$ and $\{ V_t\}_t$ and full feedback. Then, for any algorithm $\mathcal A$, there exists a sequence $\{ (M_t,V_t)\}_{t=1}^T$ where  expected regret $R_T \geq \frac{T}{6}$. 
\end{restatable}
The proof of Theorem~\ref{th:adv_LB} employs the well-known Yao's minimax principle~\citep{4567946}. Specifically, we build a distribution over problem instances in which the relevant price interval shrinks from either the left or the right at random in each round. The learner cannot anticipate which part of the interval will remain, whereas a single fixed action chosen in hindsight can be selected using the final interval. This uncertainty creates a constant gap between the learner's expected utility and that of the hindsight benchmark in each round, yielding a linear regret lower bound. The full proof is given in Appendix~\ref{app:LB_fullfeedMM}.

\bibliography{example_paper}
\bibliographystyle{plainnat}

\appendix
\section*{Appendix}
The Appendix is organized as follows:
\begin{itemize}
    \item In Appendix~\ref{app:advstoc_fullfeed}, we provide the omitted proofs and lemmas of Section~\ref{sec:advstoc_fullfeed}.
    \item In Appendix~\ref{app:advstoc_LOB}, we provide the omitted proofs and lemmas of Section~\ref{sec:advstoc_LOB}.
    \item In Appendix~\ref{app:LB_fullfeedMM}, we prove the lower bound stated in Section~\ref{sec:LB_fullfeedMM}.
\end{itemize}

\section{Warm up}\label{app:advstoc_fullfeed}

We start by proving that, by ordering the bid and ask prices, the mean utility cannot reduce.

\lemmaOderTheGame*
\begin{proof}
If $X_t=Y_t$ then the statement is trivially proved, so we assume $X_t \neq Y_t$.
If $X_t < Y_t$, then $u_t(Y_t,X_t) = u^B_t(B_t)+u_t^A(A_t)$ and $(X_t-Y_t)_+=0$ the claim is true. 

If instead $X_t > Y_t$ and $V_t \neq \{Y_t,X_t\}$, then we have that for each case $V_t < Y_t <X_t$, $V_t \in (Y_t,X_t)$ and $ Y_t <X_t < V_t$ we have an improvement of $(X_t-Y_t)_+$ in the utility by swapping $X_t$ with $Y_t$. 
For example, if $V_t < Y_t <X_t$, then $u_t(X_t,Y_t)=u_t^A(X_t) + u_t^B(Y_t)=(M_t-Y_t) = u_t(Y_t,X_t) + (X_t-Y_t)_+$ and a similar reasoning can be applied to the other cases.

In the cases where $V_t \in \{A_t,B_t\}$, if $V_t = X_t$ and $X_t > Y_t$, then $u_t(Y_t,X_t)=(Y_t-X_t)$ and $u_t(X_t,Y_t)=0$, so the statement still holds.
Lastly, if $V_t = Y_t$ and $X_t > Y_t$, $u_t(Y_t,X_t)=(M_t-X_t)$ and $u_t(X_t,Y_t)=(M_t-Y_t)$, by which we conclude the first part of the proof.

By observing that $(X_t-Y_t)_+ \geq 0$, we take the expected value of the previously proved equality, yielding $U_t(A_t,B_t)\geq U_t(Y_t,X_t)$. Finally, by taking the empirical mean over the sequence $\{M_t\}_{t=1}^T$, we get $U(A_t,B_t)\geq U(Y_t,X_t)$.
\end{proof}

We now turn to the intermediate Lemma~\ref{lemma:pre_regretOMDreduction}, which justifies the reduction to a finite set of actions in Algorithm~\ref{alg:Hedge_onthegrid}, by showing explicitly that there exists an almost optimal point on the grid, even though the utility is not regular.
\preregretOMDreduction*
\begin{proof}
We define a point on $(A^*,B^*) \in\gG \times \gG$ and then prove that such point has utility $2/K$ close to the optimal action $(a^*,b^*)\in \gU$. While the supremum is not necessarily attained for the utility, we select $0 <\varepsilon < \frac{1}{2K}$ and a point such that $U(a^*,b^*) \geq \sup_{(a,b)\in \gU}U(a,b)-\varepsilon$. 

Let us first notice that we can reduce to the case where $b^* \leq \mu_T \leq a^*$. 
This is because of the structure of our utility function: if $a^*< \mu_T$, we have $\left(a^*-\mu_T)(1-F(a^*)\right)\leq 0$, so we would weakly improve our revenue by substituting $a^*$ with $\mu_T$: analogously, we can argue that $b^* \leq \mu_T$
We select the point $(A^*,B^*)$ following the same order, using the same logic. 

Let $(A^*,B^*) \in\gG \times \gG$ be defined as:
\begin{align}\label{eq:discr_optexp_stocstoc}
    B^*=
\begin{cases}
\displaystyle \frac{\lceil Kb^*\rceil}{K},
&
\text{if }\frac{\lceil Kb^*\rceil}K\leq\mu_T,\\[10pt]
\displaystyle \frac{\lfloor K\mu_T\rfloor}{K},
&
\text{otherwise}.
\end{cases} \qquad
A^*=
\begin{cases}
\displaystyle \frac{\lfloor Ka^*\rfloor}{K},
&
\text{if }\frac{\lfloor Ka^*\rfloor}{K}\geq\mu_T,\\[10pt]
\displaystyle \frac{\lceil K\mu_T\rceil}{K},
&
\text{otherwise}.
\end{cases}
\end{align}

 We decouple the contributions from the ask and the bid components, only studying the bid case, as the other follows from a similar argument: we then we sum up the additional contribution of $\varepsilon$ to measure the difference with the supremum.
 Suppose first that the first case holds and
$
B^*=\frac{\lceil K b^*\rceil}{K}\leq\mu_T.
$
Since $F$ is non decreasing and bounded in $[0,1]$, $F(B^*)\geq F(b^*)$.
Therefore, $ (\mu_T-B^*)F(B^*) \geq(\mu_T-B^*)F(b^*)=(\mu_T-b^*)F(b^*)-(B^*-b^*)F(b^*)\geq (\mu_T-b^*)F(b^*)-\frac{1}{K}.$

Suppose instead that $B^*=\frac{\lfloor K\mu_T\rfloor}{K}.$ In this case the first rounding point lies above $\mu_T$, which implies
$0\leq\mu_T-b^*<\frac{1}{K}.$ Since $F(b^*)\leq1$, $
0\leq (\mu_T-b^*)F(b^*)<\frac{1}{K}.$
On the other hand, $B^*\leq\mu_T$, so $(\mu_T-B^*)F(B^*)\geq0$.
Thus again $(\mu_T-b^*)F(b^*)-(\mu_T-B^*)F(B^*)\leq\frac{1}{K}$.

Hence, in all cases, $(\mu_T-b^*)F(b^*)-(\mu_T-B^*)F(B^*)\leq\frac{1}{K}$, which concludes the argument, by summing up the two contributions from bid and ask prices.
We conclude the argument by observing that:
$$\sup_{(a,b)\in \gU}U(a,b) - U(A^*,B^*) = \sup_{(a,b)\in \gU}U(a,b) - U(a^*,b^*) + U(a^*,b^*) - U(A^*,B^*)\leq \frac{5}{2K}.$$
This concludes the proof.
\end{proof}

After proving that there exists an almost-optimal point with respect to $K$ on the grid $\gG \times \gG$, we finally turn to proving our regret bound.

\Thregretfull*
\begin{proof}
Using Lemma~\ref{lemma:pre_regretOMDreduction} and the regret decomposition in Equation~\eqref{eq:meanregret_def}, we have:
$$R_T \leq \frac{5T}{2K} + \max_{(A^*,B^*)\in \gG \times \gG} \sum_{t=1}^T \left(U_t(A^*,B^*) - U_t(A_t,B_t)\right).$$ 
We first bound the difference between the expected utility for a fixed action $(A^*,B^*)$ and the corresponding observed utility. 
We observe that:
\begin{align*}
    &\sum_{t=1}^T \left(U_t(A^*,B^*) - u_t(A^*,B^*)\right)\\ & \mkern50mu=\sum_{t=1}^T \left[(M_t - B^*) \left(F(B^*)- \ind{V_t \leq B^*} \right) + (A^*-M_t)\left( 1-F(A^*) - \ind{V_t > A^*}\right) \right].
\end{align*}
The two terms of the sum have conditional range $|M_t -B^*|$ and $|A^* - M_t|$ bounded in $1$ and form a martingale sum. 
 Thus, using the conditional Azuma-Hoeffding inequality \citep{10.2748/tmj/1178243286} and a union bound for the two terms, we find that with probability $1- \nicefrac{\delta}{3}$ the term is bounded by $ \sqrt{2T \log(\frac{3(K+1)}{\delta})}$.
 As it holds for every action $(A,B) \in \gG \times \gG$, we can select an optimal action in hindsight $(A^*,B^*) \in \argmax_{(A,B)\in \gG \times \gG} \sum_{t=1}^TU_t(A,B)$ and obtain the same bound.
Let us now consider the action $(Y_t,X_t)$ selected at each time step $t$ by the Hedge algorithms. 
By Azuma-Hoeffding inequality \citep{10.2748/tmj/1178243286} we have that with probability $1- \nicefrac{\delta}{3}$:
$$
 \sum_{t=1}^T \left(u_t(Y_t,X_t) - U_t(Y_t,X_t) \right)\leq \sqrt{2T \log(\frac{3}{\delta})}.
$$

We can rewrite the regret on the grid actions as: 
\begin{align*}
    \sum_{t=1}^T &\Big(U_t(A^*,B^*) - U_t(A_t,B_t)\Big) \leq \sum_{t=1}^T \Big( U_t(A^*,B^*) - U_t(Y_t,X_t)\Big) \\
    & = \sum_{t=1}^T \Big[u_t(A^*, B^*) - u_t(Y_t, X_t)\Big] +  \sum_{t=1}^T\Big[U_t(A^*,B^*) - u_t(A^*,B^*)\Big] + \sum_{t=1}^T\Big[u_t(Y_t,X_t) - U_t(Y_t, X_t)\Big],
\end{align*}
where the inequality holds by applying Lemma~\ref{lemma:order_the_game}, while the equality holds because the utility can be decomposed into separate components. 
We notice here that we use in the proof the points $\{(X_t,Y_t)\}_{t=1}^T$ selected in the original instances of Hedge for the bid and ask, so that the standard concentration argument can be applied, for example using \citet{NIPS2014_24402144}. 
Thus, observing that both instances of Hedge have the same learning rate $\eta>0$, we have that with probability $1 - \nicefrac{\delta}{3}$, using the decomposition into two separate problems: 
\begin{align*}
    \sum_{t=1}^T \Big[u_t(A^*, B^*) - u_t(Y_t, X_t)&\Big] = \sum_{t=1}^T \Big[u_t^A(A^*) - u_t^A(Y_t)\Big] + \sum_{t=1}^T \Big[ u_t^B(B^*) - u_t^B(X_t)\Big]\\
    &\leq \frac{2\log(K+1)}{\eta}+ \eta\, T + 2\sqrt{T \log\left( \frac{3}{\delta}\right)} ,
\end{align*} 

Building on that, by choosing $K = \lceil \sqrt{T}\rceil$ and $\eta= \sqrt{\frac{\log(T)}{T}}$, we get the second result $R_T = \gO (\sqrt{T \log(T)} + \sqrt{T \log(1/\delta)}+ \log(1/\delta))$ with probability $1-\delta$ by summing the failure probabilities, which concludes the argument. 
\end{proof}

\section{Optimal bounds for LOB}\label{app:advstoc_LOB}

\partfeedreconstruction*
\begin{proof}
    If $V_t \in (B_t,A_t]$ the statement follows trivially by the fact that we observe $(M_t,V_t)$ entirely. 
 
    If a purchase occurred and $x \in \gG_B$, then $V_t \leq B_t \leq x$: therefore $u^B_t(x)=M_t-x$. 
    If instead we consider a point on the other grid $y \in \gG_A \cap [B_t,A_t]$, then $V_t \leq B_t \leq y$ and $u_t^A(y)=0$: therefore we can again reconstruct the feedback.
    The same reasoning can be mirrored in case of a sale, thus we conclude the proof.  
    We notice here that we can reconstruct the revenue independently of the $V_t$ distribution, as it only depends on the structure of the utility and the feedback. 
\end{proof}

After proving the consistency of our estimators, we now move on to the proof of Lemma~\ref{lemma:bound_logweights}.
This lemma presents the first step for bounding in the regret analysis. 

\boundlogweights*
\begin{proof}
    For the sake of this proof, we fix a unique indexing on the union of the two grids $\mathcal{G}_A \cup \mathcal G_B$, with a total of $2(K+1)$ points, denoted as $z_1 \leq \dots \leq z_N$. 
    We can therefore call the probability weights with the unique name of $w_t$ for simplicity of notation, where $w_t(z_i)=w_t^B(z_i)$ if $z_i \in \gG_B$ and $w_t(z_i)=w_t^A(z_i)$ if $z_i \in \gG_A$.
     
     Let $S_t(z_i) = \frac{1}{2}\sum_{j \leq i}{w_t(z_j)}$ be the normalized sum of the probabilities in each point on the grid until point $z_i \in \gG_A \cup \gG_B$. 
    Recall Equation~\eqref{eq:Qderivation}: we observe that, selecting a point $z_i \in \gG_A \cup \gG_B$, it holds that $\sum_{j < i}w_t(z_j) = 2S_{t}(z_{i-1})$ and $\sum_{j > i}w_t(z_j) = 2 (1-S_t(z_i))$ with our new notation.  
    Thus we have that, from Identity~\eqref{eq:Qderivation}:
    \begin{align*}
            Q_t(z_i) &\geq 1-S_{t}^2(z_{i-1})- (1-S_{t}(z_i))^2
            \\
            &= 2S_{t}(z_i)(1-S_{t}(z_{i-1}))-\left(\frac{w_t(z_i)}{2}\right)^2
            \\&\geq S_{t}(z_i)(1-S_{t}(z_{i-1})).
    \end{align*}
     By using the standard analysis inequality $ab \leq \nicefrac{(a+b)^2}{4}$ and the fact that: \[S_t(z_i)\left(1-S_{t}(z_{i-1}) \right) \geq \left(\frac{w_t(z_i)}{2}\right)^2,\] since $\left(\frac{w_t(z_i)}{2}\right) = S_t(z_i)-S_t(z_{i-1})$ is no larger than both $S_t(z_i)$ and $1-S_t(z_{i-1})$.

   We use these inequalities to construct our bound on the importance sampling term:
    \begin{align*}
            \sum_{i=1}^N \frac{w_t(z_i)}{Q_t(z_i) + \beta} &\leq \sum_{i=1}^N \frac{w_t(z_i)}{S_t(z_i)(1-S_t(z_{i-1})) + \beta} \\
            &\leq \sum_{i=1}^N \frac{w_t(z_i)}{S_t(z_i) + \beta} + \sum_{i=1}^N \frac{w_t(z_i)}{1-S_{t}(z_{i-1}) + \beta} 
            \\
            &=2\sum_{i=1}^N \frac{S_t(z_i) -  S_t(z_{i-1})}{S_t(z_i) + \beta} + 2\sum_{i=1}^N \frac{S_t(z_i) -  S_t(z_{i-1})}{1-S_{t}(z_{i-1}) + \beta}
            \\
            &\leq 4\int_{0}^1 \frac{d\,S}{S+ \beta} = 4\left(\log\left(\frac{1+\beta}{\beta}\right) \right).
    \end{align*}

  Which holds by applying definitions and using that $1/(ab)\leq 1/a + 1/b$ with $a+b\geq 1$ and standard inequalities from analysis. 

     We define for the following proofs the factor in Equation~\eqref{eq:bound_logweights} as: $$D\coloneqq 4\left(\log\left(\frac{1+\beta}{\beta}\right) \right),$$ to simplify later derivations. 
\end{proof}

Before proving Theorem~\ref{th:partialregrettheorem}, let us state some intermediate steps Lemma~\ref{lemma:pre_expregret}, Lemma~\ref{lemma:HP_diffloss} and Lemma~\ref{lemma:aggr_auxvariables}, providing high probability bounds on the conditional mean and variance of the estimated losses at each time step, given the data from the previous rounds $\gH_{t-1}=\{(A_1,B_1),(V_1,M_1), \dots, (V_{t-1},M_{t-1})\}$. 

In the following Lemma~\ref{lemma:pre_expregret}, we bound the cumulative difference between the estimated loss of the two current algorithms, denoted as $Z_t$, and that of any fixed bid--ask pair $(a,b)\in\gG_A\times\gG_B$, up to the usual logarithmic and second-order terms, denoted as $Y_t$.

\begin{lemma}\label{lemma:pre_expregret}
Let $Z_t$, $Y_t$ be defined as:
\[
Z_t = \sum_{x \in \gG_s,\, s\in \{ A,B\}} w^s_t(x)\hat \ell_t^s(x), \qquad Y_t = \sum_{x \in \gG_s,\, s\in \{ A,B\}} w^s_t(x)(\hat \ell_t^s(x))^2.
\]

Then we have that, for any $b \in \gG_B,\, a \in \gG_A$:
\[
\sum_{t=1}^T\left( Z_t - \left(\hat \ell_t^B(b)+  \hat \ell_t^A(a) \right) \right) \leq \frac{2\log(K+1)}{\eta} + \frac{\eta}{2} \sum_{t=1}^T Y_t.
\]   
\end{lemma}

\begin{proof}
    The proof follows by manipulation of the weights $W^s_t$.
    Fix a time step $t\leq T$ and let $s \in \{A,B\}$.  
    First, consider that:
    \[\frac{ W_{t+1}^s(x) }{\sum_{y \in \mathcal G_s} W^s_t(y)}= \frac{W_{t}^s(x) }{\sum_{y \in \mathcal G_s} W^s_t(y)} \exp(-\eta \hat \ell_t^s(x))=w^s_t(x) \exp(-\eta \hat \ell_t^s(x)).\]
    Therefore, if we sum the equality above over the grid points $x \in \gG_s$, we have:
    \begin{align*}
    \frac{ \sum_{x \in \mathcal G_s} W^s_{t+1}(x) }{\sum_{y \in \mathcal G_s} W^s_t(y)} & = \sum_{x \in \mathcal G_s} w^s_t(x) \exp(-\eta \hat \ell_t^s(x)) \\& \leq \sum_{x \in \mathcal G_s} w^s_t(x)\left( 1 -\eta \hat \ell_t^s(x) + \frac{\eta^2}{2}(\hat \ell_t^s(x) )^2\right)\\
    &=  1 - \eta \sum_{x \in \mathcal G_s} w^s_t(x)\hat \ell_t^s(x)  + \frac{\eta^2}{2}\sum_{x \in \mathcal G_s} w^s_t(x)(\hat \ell_t^s(x) )^2,
    \end{align*}
    where we used that $\exp(-x) \leq 1 - x + \frac{x^2}{2}$.
    After taking the logarithm on both sides and employing that $\log(1+x) \leq x$:
    \[
    \log\left( \frac{ \sum_{x \in \mathcal G_s} W^s_{t+1}(x)}{\sum_{y \in \mathcal G_s} W^s_t(y)}\right) \leq  - \eta \sum_{x \in \mathcal G_s} w^s_t(x)\hat \ell_t^s(x)  + \frac{\eta^2}{2}\sum_{x \in \mathcal G_s} w^s_t(x)(\hat \ell_t^s(x) )^2, 
    \]
    and by summing over the time steps $t= 1, \dots, T$ and $s\in \{A,B\}$, and telescoping, we deduce:
    \[
        \log\left(\sum_{y \in \mathcal G_A} W^A_{T+1}(y) \right) + \log\left(\sum_{y \in \mathcal G_B} W^B_{T+1}(y) \right) - 2\log(K+1) \leq \sum_{t=1}^T \left(-\eta Z_t + \frac{\eta^2}{2}Y_t \right). 
    \]
    where we used $\sum_{y \in \mathcal G_s} W^s_1(y)=K+1,\,\, s \in \{A,B\}$ and the definitions of $Z_t$ and $Y_t$.
    We conclude the argument observing that, fixing an $s\in \{A,B\}$ and $x \in \mathcal G_s$, $\sum_{y \in \mathcal G_s} W^s_{T+1}(y) \geq W_{T+1}^s(x) = \exp\left( -\eta \sum_{t=1}^T \hat \ell^s_t(x)\right)$, since it represents the last update, and substituting this inequality into the last step of the previous chain of inequalities. 
\end{proof}

We now provide a high probability bound on the difference between the estimated losses $\hat \ell$ and the real $\ell$. 

\begin{lemma}\label{lemma:HP_diffloss}
Consider the realized mean loss a time $t$ and point $x \in \gG_B$, defined as $\bar \ell_t^B(x) \coloneqq \frac{1 - (M_t -x)F(x)}{2}$ and for $y \in \gG_A$ as $\bar \ell_t^A(y) \coloneqq \frac{1 - (y - M_t)(1-F(y))}{2}$.

    For any $x \in \gG_s$, with $s \in \{A,B\}$ and for all $u > 0$ it holds that:
    \[
    \mathbb P\left( \sum_{t=1}^T \left(  \hat \ell_t^s(x) - \bar \ell_t^s(x) \right) \geq \frac{u}{\beta}\right) \leq \exp(-u).
    \]
    Hence, setting $\delta\in (0,1)$, with probability $1-\delta/3$ we have that for any $z \in \mathcal{G}_s$:
    \[
    \sum_{t=1}^T  \left(\hat \ell_t^s(x) - \bar \ell_t^s(x) \right) \leq \frac{1}{\beta}\log\left( \frac{6(K+1)}{\delta}\right).
    \]
\end{lemma}

\begin{proof}
    Let $x \in \gG_s$ and fix a time step $t$.
    We observe that $O_t(x)$ is distributed as a Bernoulli with parameter $Q_t(x)$, given the history $\gH_{t-1}$ and loss $\ell^s_t(x)$. 
    By taking the exponential of the estimator $\hat \ell_t^s(x)$, 
    we can therefore bound:
    \begin{align*}
            \mathbb E_{O_t(x)} \left[\exp \left(\frac{\beta O_t(x)\ell^s_t(x) }{Q_t(x)+ \beta}\right) \mid \mathcal H_{t-1}, \ell^s_t(x) \right] &= 1-Q_t(x)+ Q_t(x) \exp \left(\frac{\beta\, \ell^s_t(x) }{Q_t(x) + \beta}\right)\\
            &\leq 1 + \frac{\beta \,Q_t(x)\, \ell^s_t(x)}{Q_t(x) + \beta - \beta \ell^s_t(x)} \\&\leq 1 + \beta \,\ell^s_t(x). 
    \end{align*}
Where we used that, for $x \in [0,1)$, $\exp(x) \leq 1 + \frac{x}{1-x}$, then that $\frac{,Q_t(x)}{Q_t(x) + \beta - \beta \ell^s_t(x)} \leq 1$. 

By the tower rule on $\ell_t^s(z)$, we can rewrite the expected value $\mathbb E \left[ \exp \left(\hat \ell^s_t(x)\right)\mid \mathcal H_{t-1} \right]$:
\begin{align*}
    \mathbb E \left[\mathbb E_{O_t(x)} \left[\exp \left(\beta \hat \ell_t^s(x)\right) \mid \mathcal H_{t-1}, \ell_t^s(x) \right] \mid \mathcal H_{t-1}\right] &\leq 1+  \beta \, \bar \ell_t^s(x) \\
    &\leq \exp \left(\beta\,  \bar \ell_t^s(x) \, \right),
\end{align*}
where the last inequality follows by $1+x \leq \exp(x)$.

Hence, the process $\exp\left(\beta \sum_{\tau=1}^t \left( \hat \ell_\tau^s(x)- \bar\ell_\tau^s(x) \right)\right)$ is a nonnegative supermartingale that has an initial value of 1. 

The first part of the thesis is deduced by applying Markov's inequality to the difference of the sums $\sum_{t=1}^T \left( \hat \ell_t^s(x)- \bar \ell_t^s(x)\right)$ and applying the deduced inequality in terms of the exponential.
That is, 
\begin{align*}
    &\mathbb P\left(\sum_{t=1}^T \left( \hat \ell_t^s(x)-\bar\ell_t^s(x)\right) \geq \frac{u}{\beta} \right)= \\
    &= \mathbb P\left( \exp \left( \beta \,\sum_{t=1}^T \left( \hat \ell_t^s(x)- \bar\ell_t^s(x) \right)\right) \geq \exp(u) \right) \leq \frac{1}{\exp(u)}.
\end{align*}
The second part of the claim is obtained building on the first, by a union bound on the $2(K+1)$ total grid points and where we select $u=\log(\frac{6(K+1)}{\delta})$.
\end{proof}
\begin{lemma}\label{lemma:aggr_auxvariables}
    Consider $\eta \in (0,\beta]$ and recall the definition of $Z_t$ and $Y_t$ given in Lemma~\ref{lemma:pre_expregret}.
    Then, it holds that $0 \leq Z_t - \frac{\eta}{2}Y_t \leq D,$ and $\mathbb E[Z_t - \frac{\eta}{2}Y_t \mid \mathcal H_{t-1}] \leq 2$.
    Moreover:
    \begin{align*}
         \sum_{x \in \gG_s,s\in \{ A,B\}} w^s_t(x)\, \bar\ell_t^s(x)- \mathbb E\left[Z_t - \frac{\eta}{2}Y_t \mid \mathcal H_{t-1}\right] \leq \left( \beta + \frac{\eta}{2}\right)D.
    \end{align*}
Furthermore, for any $u>0$:
\[
\mathbb P\left( \sum_{t=1}^T \left( \mathbb E\left[Z_t - \frac{\eta}{2}Y_t \mid \mathcal H_{t-1} \right] - \left( Z_t - \frac{\eta}{2}Y_t\right)\right) > 2 \sqrt{DTu}\right) \leq \exp(-u),
\]
and for every $x \in \mathcal G_B$ and $y \in \mathcal G_A$ it holds that:
\[
\sum_{t=1}^T \left(Z_t - \frac{\eta}{2}Y_t \right)- \sum_{t=1}^T \left(\hat \ell_t^B(x)+\hat \ell_t^A(y) \right) \leq \frac{2\log(K+1)}{\eta}.
\]
\end{lemma}

\begin{proof}

    We first notice that:
    $$
     \mathbb E[ \bar \ell_t^A(A_t) + \bar \ell_t^B(B_t) \mid \mathcal H_{t-1}] \leq \mathbb E[ \bar \ell_t^A(Y_t) + \bar\ell_t^B(X_t) \mid \mathcal H_{t-1}]=\sum_{x \in \mathcal G_s,\, s\in \{A,B \}} w^s_t(x) \bar\ell_t^s(x).
    $$ 
    The inequality derives from the fact that we might swap the prices $(Y_t,X_t)$ initially selected by Algorithm~\ref{alg:MM_partialfeed_exp3} using Lemma~\ref{lemma:order_the_game}, for which we weakly decrease the loss if in the step $t$ we had swapped the prices. 

    We consider the expected value for the losses estimates. 
    It holds, by definition:
    \begin{align*}
        \mathbb E[ \hat\ell_t^s(x)\mid \mathcal H_{t-1}] &= \frac{Q_t(x)}{Q_t(x) + \beta}\bar\ell_t^s(x), \\
        \mathbb E[ (\hat\ell_t^s(x))^2\mid \mathcal H_{t-1}] &= \frac{Q_t(x)}{(Q_t(x) + \beta)^2} \mathbb E[ \ell_t^s(x)^2 \mid \mathcal{H}_{t-1}] \leq \frac{1}{Q_t(x) + \beta}.
    \end{align*}
    Thus, follows that $\sum_{x \in \gG_s,s\in \{ A,B\}} w^s_t(x)\bar\ell_t^s(x) - \mathbb E[ Z_t \mid \mathcal H_{t-1}] \leq \beta D$ and $\mathbb E[ Y_t \mid \mathcal H_{t-1}]\leq D$, from Lemma~\ref{lemma:bound_logweights}.
    
    Notice that, for any time step $t$ and $x \in \gG_s$, $\hat \ell_t^s(x) \in [0,\frac{1}{\beta}]$: so
    $0 \leq  \hat \ell_t^s(x)\left(1-\frac{\eta}{2}\hat \ell_t^s(x)\right) \leq \hat \ell_t^s(x)$.
    Therefore, by recalling Lemma~\ref{lemma:bound_logweights}, we have that $0 \leq Z_t - \frac{\eta}{2}Y_t  \leq Z_t \leq D$. 
    By taking the expectation: 
   $$
   \mathbb E\left[Z_t - \frac{\eta}{2}Y_t\mid \mathcal H_{t-1}\right] \leq  \mathbb E\left[Z_t \mid \mathcal H_{t-1}\right] \leq \sum_{x \in \gG_s,s\in \{ A,B\}} w^s_t(x) \bar\ell_t^s(x) \leq 2.
   $$
   
    We deduce the second claim from Lemma~\ref{lemma:pre_expregret}. 
    We also notice that
    $\mathbb E[\left(Z_t - \frac{\eta}{2}Y_t \right)^2 \mid \mathcal H_{t-1}] \leq D \mathbb E[Z_t - \frac{\eta}{2}Y_t\mid \mathcal H_{t-1}] \leq 2D$.

    Therefore, $\forall \lambda \geq 0$ and using $\exp(-x) \leq 1 -x+ \frac{x^2}{2}$ and $1+x \leq \exp(x)$, we have that:
    \begin{align*}
            &\mathbb E\left[\exp \left( \lambda \left( \mathbb E \left[Z_t - \frac{\eta}{2}Y_t \mid \mathcal H_{t-1} \right] - \left(Z_t - \frac{\eta}{2}Y_t\right) \right)\right) \mid  \mathcal H_{t-1}\right] \\
            &\leq  \exp\left(\lambda \, \mathbb E\left[Z_t - \frac{\eta}{2}Y_t \mid \mathcal H_{t-1}\right]\right) \left(1- \lambda \mathbb E\left[Z_t - \frac{\eta}{2}Y_t \mid \mathcal H_{t-1} \right] + \frac{\lambda^2}{2} \mathbb E\left[\left(Z_t - \frac{\eta}{2}Y_t \right)^2 \mid \mathcal H_{t-1}\right] \right)\\
            & \leq \exp\left(\frac{\lambda^2}{2} \mathbb E\left[\left(Z_t - \frac{\eta}{2}Y_t \right)^2 \mid \mathcal H_{t-1} \right]\right) \\&\leq \exp(\lambda^2 D). 
    \end{align*}
    By summing over the $t$, as the sequence is a supermartingale with expectation bounded by 1 and applying Markov's inequality with $\lambda=\sqrt{u/TD}$ we deduce the fourth claim, while the last claim is a reordering of Lemma~\ref{lemma:pre_expregret}.
\end{proof}

We finally prove the full Theorem~\ref{th:partialregrettheorem} statement, using the results from the previous Lemma.  
\partialregrettheorem*

\begin{proof}
Recall $Z_t,Y_t$ as defined in Lemma~\ref{lemma:pre_expregret} and the observations regarding their conditional expected value in Lemma~\ref{lemma:aggr_auxvariables}. 
 We first observe that by conditional Hoeffding's bound for all time steps $t$:
    \[
        \mathbb E \left[\exp\left(\lambda( \bar\ell_t^A(Y_t) + \bar\ell_t^B(X_t) - \sum_{x \in \mathcal G_s,\, s\in \{A,B \}} w^s_t(x) \bar\ell_t^s(x)\right) \mid \mathcal H_{t-1} \right] \leq \exp(\lambda^2/2).
    \]
    Then with a union bound and Markov inequality:
    \[
    \mathbb P\left(\sum_{t=1}^T \left( \bar\ell_t^A(Y_t) + \bar\ell_t^B(X_t)-\sum_{x \in \mathcal G_s,\, s\in \{A,B \}} w^s_t(x) \bar\ell_t^s(x) \right) \geq \sqrt{2Tu}\right) \leq \exp(-u).
    \]
    
    To get the bound, we use the results in the previous Lemma~\ref{lemma:pre_expregret}, Lemma~\ref{lemma:HP_diffloss} and Lemma~\ref{lemma:aggr_auxvariables} by adding and subtracting the terms defined there.

    We consider as a comparator action $(a,b) \in \gG_A \times \gG_B$. 
    This is justified as we notice that a similar construction to the one in Lemma~\ref{lemma:pre_regretOMDreduction} repeated in the case of $\gG_A \times \gG_B$ establishes that we can reduce the problem to our discrete set at the total cost of $\frac{5T}{K}$, accounting for the additional constants in the loss. 

    In particular fix any pair $(a,b) \in \gG_A \times \gG_B$: with probability at least $1-\delta$, which derives from summing up the failure probabilities from the previous lemmas:
    \begin{align*}
        \Bigg[\sum_{t=1}^T &\Bigg(\bar\ell_t^A(A_t) +\bar \ell_t^B(B_t)- \bar\ell_t^A(a) - \bar\ell^B_t(b) \Bigg) \Bigg] 
        \\ 
       &\leq  \sum_{t=1}^T \left( \bar\ell_t^A(Y_t) + \bar\ell_t^B(X_t)- \sum_{x \in \mathcal G_s,\, s\in \{A,B \}} w^s_t(x) \bar\ell_t^s(x) \right)
       \\
      & \mkern20mu+  \sum_{t=1}^T \left(\sum_{x \in \gG_s,s \in \{A,B \}} w^s_t(x) \bar\ell_t^s(x) - \mathbb E\left[Z_t - \frac{\eta}{2}Y_t \mid \mathcal H_{t-1}\right]\right) 
      \\
      & \mkern20mu+ \sum_{t=1}^T \left(\mathbb E\left[Z_t - \frac{\eta}{2}Y_t \mid \mathcal H_{t-1}\right]- Z_t + \frac{\eta}{2}Y_t\right)
      \\
      & \mkern20mu+  \sum_{t=1}^T \left(Z_t - \frac{\eta}{2}Y_t -\hat \ell_t^A(a) - \hat \ell_t^B(b) \right) + \sum_{t=1}^T \left(\hat \ell_t^A(a) + \hat \ell_t^B(b) - \bar\ell_t^A(a) - \bar \ell_t^B(b)\right) 
      \\
      & \leq  \sqrt{2T \log(3/\delta)} + \left( \beta +\frac{\eta}{2}\right) D \, T + 2\sqrt{DT\log(3/\delta)} 
      \\
       & \mkern20mu+\frac{2\log(K+1)}{\eta} + \frac{2\log(6(K+1)/\delta)}{\beta}.
    \end{align*}
    By applying this bound to the regret in the utility with factor 2 due to the loss rescaling, we obtain the general claim.
    We also used Lemma~\ref{lemma:order_the_game} in the first inequality.

    Consequently, by selecting $K= \lceil \sqrt{T} \rceil$, $\eta = \beta = \sqrt{1/T}$ we have $D \propto \log(T)$, thus we deduce the optimal rate $R_T \leq \tilde{\gO} (\sqrt{T})$.
    In particular, with probability $1-\delta$ it holds that
   $R_T = \gO\left(\sqrt{T}\log \left(\frac{T}{\delta}\right) \right)$, which concludes the proof.   
\end{proof}

\section{Proof for the impossibility result}\label{app:LB_fullfeedMM}

In this final section, we give a proof of the unlearnability result for market making, under the favorable condition of full feedback and oblivious adversary. 
\LBforMM*
\begin{table}[t]
\caption{Upper bound on $u_t(A_t,B_t)$ for each case of the random sequence}
\label{table_LBresults}
\centering
\small
{\renewcommand{\arraystretch}{1.25}
\setlength{\tabcolsep}{8pt}
\begin{tabular}{ c | c c c c | c}
\toprule
\diagbox{\textbf{$(A_t,B_t)$}}{\textbf{$(V_t,M_t)$}}
& $(L_{t-1},1)$
& $(y,1)$
& $(x,0)$
& $(U_{t-1},0)$
& \textbf{Upper bound}
\\
\midrule
$A_t \in [L_{t-1},x), B_t \in [L_{t-1},x),$
& $1-B_t$
& $A_t-1$
& $A_t$
& $A_t$
& $\frac{1}{3}$
\\
\addlinespace[3pt]
$A_t \in [x,y), B_t \in [L_{t-1},x)$
& $1-B_t$
& $A_t-1$
& $0$
& $A_t$
& $\frac{1}{4}$ 
\\
\addlinespace[3pt]
$A_t \in [y,U_{t-1}], B_t \in [L_{t-1},x)$
& $1-B_t$
& $0$
& $0$
& $A_t$
& $\frac{1}{3}$ 
\\
\addlinespace[3pt]
$A_t \in [x,y),B_t \in [x,y)$
& $1-B_t$
&  $A_t-1$
& $-B_t$
& $A_t$
& $\frac{1}{18}$ 
\\
\addlinespace[3pt]
$ A_t \in [y,U_{t-1}], B_t \in [x,y)$
& $1-B_t$
& $0$
& $-B_t$
& $A_t$
& $\frac{1}{3}$ 
\\
\addlinespace[3pt]
$A_t \in [y,U_{t-1}], B_t \in [y,U_{t-1}]$
& $1-B_t$
& $1-B_t$
& $-B_t$
& $A_t$
& $\frac{1}{3}$ 
\\
\bottomrule
\end{tabular}}
\end{table}

\begin{proof}
    To prove the statement, we rely on Yao's minimax principle \citep{4567946}: we construct a random instance of the problem and prove that any deterministic bid-ask sequence suffers linear regret with such sequence. 

    Consider the following construction for a sequence bounded in $[0,1]$ ${(L_t,U_t)}_t$. Let $D_t=U_t-L_t$ at any given time step $t$. 
    We initialize $(L_0,U_0)=(1/3,2/3)$, then for any $t=1, \dots,T$:
    \begin{align*}
        \begin{cases}
            (L_{t},U_t) = (L_{t-1}, U_{t-1}-\frac{2D_{t-1}}{3}) \qquad \text{with probability} \quad \frac{1}{2}\\
            (L_{t},U_t) = (L_{t-1} + \frac{2D_{t-1}}{3}, U_{t-1}) \qquad \text{with probability} \quad \frac{1}{2}
        \end{cases}
    \end{align*}

    Given this sequence, we construct the instance of prices $\{(V_t,M_t)\}_t$ as follows:
     \begin{align*}
        \begin{cases}
            (V_{t},M_t) = (L_t,1), \qquad \text{with probability} \quad \frac{1}{2},\\
            (V_{t},M_t) = (U_t,0), \qquad \text{with probability} \quad \frac{1}{2},
        \end{cases}
    \end{align*}
    where the sequence is drawn independently of the previous one. 

    We now consider an optimal in hindsight action for such instance of the problem. 
    We notice that there exists a $(a^*,b^*)$ where $a^*,b^* \in [L_T,U_T]$ such that the learner always wins. Defining a constant $\varepsilon \leq \frac{D_T}{2}$, we can take such action as $(a^*,b^*)=(U_T - \varepsilon, L_T)$. 
    The total expected revenue with respect to the randomization of the environment is: 
    \[
    \mathbb E \left[\sum_{t=1}^T u_t(a^*,b^*) \right] = T\frac{1-L_T + U_T - \varepsilon}{2} = T\frac{1 + D_T - \varepsilon}{2} \geq \frac{T}{2} + \frac{T \,D_T}{4}
    \]
    where we imposed $\varepsilon=\frac{D_T}{2}$ and $D_T = \frac{1}{3^{T+1}}$ and vanishes as $T$ grows.

    We now turn to the expected revenue for any deterministic sequence.
    For any time step $A_t > U_t$ or $B_t < L_t$ gains $0$ revenue for the corresponding utility, while having $A_t\leq L_t$ or $B_t \geq U_t$ doesn't improve utility compared to choosing $A_t = L_t$ and $B_t=U_t$. 
    With this observation, we can restrict to only considering the cases of bid-ask prices $A_t \in (L_{t-1},U_{t-1}]$ and $B_t \in [L_{t-1},U_{t-1})$.
    We have six cases to consider, depending on where the bid ask prices lie in the three intervals $(L_{t-1}, L_{t-1} + \frac{D_{t-1}}{3})$, $(L_{t-1} + \frac{D_{t-1}}{3}, L_{t-1} + \frac{2D_{t-1}}{3})$ and $(L_{t-1} + \frac{2D_{t-1}}{3}, U_{t-1})$ and given our general constraint $A_t \geq B_t$. 
    For simplicity, we call $x=L_{t-1} + \frac{D_{t-1}}{3}$ and $y= L_{t-1} + \frac{2D_{t-1}}{3}$ in Table~\ref{table_LBresults}.
    We proceed by computing the expected revenue given each event in the original sequence and summing up the resulting revenue for each case, thus obtaining the general expected revenue for a deterministic sequence of prices. 
    We summarize the results for each case with Table~\ref{table_LBresults}.

    Overall, the expected utility for a deterministic algorithm playing $(A_t,B_t)\in \mathcal U$ is $ \mathbb E \left[u_t(A_t,B_t) \right] \leq \frac{1}{3}$, given our construction.

    Therefore by Yao's principle that for any probabilistic algorithm $\{(A_t,B_t)\}_{t=1}^T$:
    \begin{align*}
        R_T \geq \frac{T}{2} - \frac{T}{3}=\frac{T}{6},
    \end{align*}
    which proves the statement, by Yao's principle.  
\end{proof}

\end{document}